\documentclass[journal]{IEEEtran}
\ifCLASSINFOpdf
\else
\fi

\usepackage{graphicx}
\usepackage{amsmath}
\usepackage{mathtools}
\usepackage{epsfig}
\usepackage{float}
\usepackage{lipsum}
\usepackage{stfloats}
\usepackage{array}
\usepackage{amssymb}
\usepackage{cite}
\usepackage{url}
\usepackage{algorithm}
\usepackage{algpseudocode}
\usepackage{multirow}
\usepackage{fancyh}
\usepackage{upgreek}
\usepackage{dsfont}
\usepackage{gensymb}
\usepackage{physics}
\usepackage{amsthm}

\usepackage{multicol, blindtext}
\usepackage{mathtools, cuted}

\usepackage{color}
\usepackage{bm}
\usepackage{booktabs}
\usepackage{tablefootnote}
\usepackage{threeparttable}
\usepackage{makecell}

\usepackage{tikz}
\usepackage{pgfmath, pgfplots}
\usetikzlibrary{3d, shapes.geometric, calc}
\usepackage{color}

\usepackage{pifont}

\def\conv{{\rm{Conv}}}

\def\bb0{{\mathbb{0}}}

\def\bb{{\boldsymbol{b}}}

\def\bw{{\boldsymbol{w}}}

\def\b0{{\boldsymbol{0}}}

\def\x{{\mathrm{x}}}

\def\b{{\mathrm{b}}}

\def\k{{\mathrm{k}}}

\def\r0{{\mathbf{0}}}

\def\bbC{{\mathbb{C}}}

\def\bbE{{\mathbb{E}}}

\def\bbQ{{\mathbb{Q}}}
\def\bbR{{\mathbb{R}}}

\def\cW{\mathcal{W}}

\def\bsf0{{\bm{\mathsf{0}}}}

\def\N0{{N_{\mathrm{0}}}}

\def\j{\mathrm{j}}

\def\bsf{{\boldsymbol{s}_\mathrm{f}}}

\newcommand{\be}{\begin{equation}}
\newcommand{\ee}{\end{equation}}
\newcommand{\bal}{\begin{align}}
\newcommand{\eal}{\end{align}}

\def\E{{\mathbb E}}

\theoremstyle{remark}
\newtheorem{theorem}{Theorem}
\newtheorem{lemma}{Lemma}

\newtheorem{corollary}{Corollary}

\normalsize
\begin{document}
%
\title{Beamforming Gain with Nonideal Phase Shifters
\thanks{H. Do and A. Lozano are with Univ. Pompeu Fabra, 08018 Barcelona (e-mail:\{heedong.do, angel.lozano\}@upf.edu).
Their work is supported the Maria de Maeztu Units of Excellence Programme CEX2021-001195-M funded by MICIU/AEI/10.13039/501100011033, and by the Departament de Recerca i Universitats de la Generalitat de Catalunya.
}
}

\author{\IEEEauthorblockN{Heedong~Do},
 {\it Member,~IEEE},
 \and
 \IEEEauthorblockN{Angel~Lozano},
{\it Fellow,~IEEE}
\vspace{-4mm}
}
\maketitle



%


\maketitle

\begin{abstract}

This research sets forth a universal framework to characterize the beamforming gain achievable with
arbitrarily nonideal phase shifters.
Precisely, the maximum possible shortfall relative to the gain attainable with ideal phase shifters is established. Such shortfall is shown to be fundamentally determined by the perimeter of the convex hull of the set of feasible beamforming coefficients on the complex plane. This result holds regardless of whether the beamforming is at the transmitter, at the receiver, or at a reconfigurable intelligent surface. 
In i.i.d. fading channels, the shortfall hardens to the maximum possible shortfall as the number of antennas grows.
\end{abstract}



%
\IEEEpeerreviewmaketitle

\section{Introduction}



Beamforming is a long-standing technique in wireless communication, instrumental to increase the received signal strength and/or extend the transmission range \cite[Sec. 5.3]{heathlozano2018foundations}. 
The relevance of beamforming has only grown over time, as operating frequencies increase and
the array dimensionalities surge to counter the associated propagation losses.
In 6G, 
base stations could feature well over a thousand antennas in their arrays \cite{wang2024tutorial}.
Beamforming will be a chief ingredient, and a complete understanding of its performance is thus of paramount importance.

The scope of this paper extends to any beamformer, analog or digital, at transmitter or receiver, 
that is implemented by shifting the signal's phase at each antenna.
For the sake of performance analysis, phase shifters are usually regarded as \text{ideal}, meaning that they can manipulate the phase unrestrictedly, with an infinite resolution, and without altering the signal's magnitude.
This ideality, however, does not hold in actuality because:
\begin{itemize}
    \item For digital beamformers, the resolution is necessarily coarse \cite[Table 2-5]{chakraborty2017paradigm}, and the maximum phase change
    is limited \cite[Table 6]{chakraborty2017paradigm}.
    \item For analog beamformers, while the resolution is infinite,
    the range is again limited \cite[p. 27]{chakraborty2017paradigm}.
    \item Insertion losses are not independent of the phase shift, which couples the phase and magnitude responses \cite{fakharzadeh2008effect}.
\end{itemize}
These and other non-idealities can be captured by defining a set (continuous or discrete) of feasible beamforming coefficients, and the question is then to quantify
the beamforming gain when the coefficients are restricted to being drawn from this set. This paper tackles this question, with the key findings being that:
\begin{itemize}
    \item For any given channel, the shortfall relative to ideal phase shifting 
    of any arbitrary set of feasible beamforming coefficients is bounded.
    \item This worst-case shortfall is intimately related to the perimeter of the convex hull of the set of beamforming coefficients on the complex plane.
    \item In i.i.d. fading channels, as the number of antennas grows large the shortfall comes to equal this worse-case value.
\end{itemize}

The manuscript is organized as follows. The main result bounding the shortfall relative to ideal phase shifting is presented immediately in Sec.~\ref{sec_main_result}, with its derivation following in Sec.~\ref{sec_derivation}. Subsequently, Sec.~\ref{geometric_proof} provides an alternative derivation that is geometric in nature, and Sec.~\ref{sec_tightening}
shows how the main result can be tightened slightly at the expense of generality. Finally, Sec.~\ref{sec_conclusion} concludes the paper.

\section{Main Result}
\label{sec_main_result}

Consider a connection between an $N$-antenna array and a single antenna. Within the array, each antenna is able to apply a beamforming coefficient drawn from a compact set
\begin{align}
    \cW \subset \{w\in\bbC: |w|\leq 1\}. \label{phase_shifter_set}
\end{align}
As a special case, an ideal phase shifter would correspond to
\begin{align}
    \cW_{\text{ideal}} = \{w\in\bbC: |w|= 1\}. \label{ideal_phase_shifter}
\end{align}
Common sets, both continuous and discrete, are listed in Table~\ref{table:phase_shifter_sets}. In particular, and for future reference, the regular $M$-gon is denoted by
\begin{align}
    \cW_M \equiv \left\{ e^{j2\pi \frac{0}{M}}, e^{j2\pi \frac{1}{M}}, \ldots, e^{j2\pi \frac{M-1}{M}}\right\}. \label{regular_polygon}
\end{align}

Let $h_n$ be the channel coefficient at the $n$th array antenna, such that
 \begin{align}
     h_n = |h_n|e^{j\theta_n},
 \end{align}
and let $w_n$ be the beamforming coefficient at the $n$th antenna. 
The beamforming gain is given by
\begin{align}
    \bigg| \sum_n w_n h_n \bigg|^2 , \label{beamforming_gain}
\end{align}
suitably normalized depending on the pertinent power constraint, which in turn depends on whether the beamforming is at a transmitter or receiver.
For the sake of compactness,
\begin{align}
    g(\bw) = \bigg| \sum_n w_n h_n \bigg| \label{objective}
\end{align}
is henceforth used in lieu of \eqref{beamforming_gain}, with $\bw\equiv (w_1,\ldots, w_N)$. 

For ideal phase shifters abiding by \eqref{ideal_phase_shifter}, the maximum value of \eqref{objective} is, from the triangle inequality,
\begin{align}
   g_{\text{ideal}} = \sum_n |h_n| \label{ideal_beamforming_gain}
\end{align}
achieved by $w_n = e^{-j\theta_n}$. 

For arbitrary phase shifters, one can simply round each $w_n = e^{-j\theta_n}$ to
$\cW$ \cite{mailloux1984array, liang2014low, yang2017study, do2023line, loyka2025irs}.
%
%
This simplicity, though, comes with a penalty in beamforming gain.
Such a penalty can be prevented by designs that, rather than rounding to $\cW$, are based on it from the outset \cite{haupt1997phase, ismail2010array, goudos2017antenna, leino2020beam, madani2021practical,narula1998efficient, heath1998simple, love2003equal, murthy2007quantization, tsai2009transmit}. 
In fact, it turns out that the optimum $\bw$ for any arbitrary $\cW$ and $h_1,\ldots,h_N$ as well as the corresponding 
\begin{align}
   g_\cW = \max_{\bw \in \cW^N} & \; \bigg| \sum_n w_n h_n \bigg| \label{original_problem}
\end{align}
can be efficiently computed without exhaustively searching over $\cW^N$ \cite{mackenthun1994fast, sweldens2001fast, motedayenaval2003polynomial, alevizos2016log, deng2019mmwave, sanchez2021optimal, zhang2022configuring, ren2023ieee, vardakis2023intelligently, pekcan2024achieving, sanjay2024optimum,do2026optimum}.

This paper
uncovers a universal relationship between this optimum value with arbitrary phase shifters, $g_\cW$, and its counterpart with ideal phase shifters, $g_\text{ideal}$.
Specifically, it is established that
\begin{align}
g_\cW \geq \text{constant} \cdot g_\text{ideal} \label{main_result}
\end{align}
with
\begin{align}
  \text{constant} = \frac{\text{perimeter of }\conv\,\cW}{2\pi}, \label{best_constant}
\end{align}
where $\conv(\cdot)$ denotes the convex hull of a set.
The shortfall in beamforming gain (in dB) associated with nonideal phase shifting is therefore, at most,
\begin{align}
20 \log_{10} \! \left( \frac{\text{perimeter of }\conv\,\cW}{2\pi} \right)  . \label{constant_in_dB}
\end{align}
Note that, if $\cW$ does not reside on the unit circle, this shortfall subsumes a power loss
\begin{align}
    10 \log_{10} \! \left( \frac{1}{N} \sum_n |w_n|^2 \right) .
\end{align}
Conversely, if $\cW$ does reside on the unit circle, then the shortfall is caused entirely by the limited range and resolution of the phase shifting.

\begin{table*}
\setlength{\tabcolsep}{3pt}
\renewcommand{\arraystretch}{1.3}
\centering
\begin{threeparttable}
\caption{Common sets $\cW$ with the shaded regions indicating the corresponding convex hulls}
\label{table:phase_shifter_sets}
\begin{tabular}{l c c c c c c c} 
\toprule
\raisebox{6.5em}{Set} &
\begin{tikzpicture}[>=stealth]
\begin{scope}[scale = 0.6]
    \clip (-1.7,-1.7) rectangle (1.7,1.7);
    \draw[line width=0.5pt, ->] (-1.7,0) -- (1.7,0) node (X) [right,xshift = -0.5 cm, yshift = -0.2 cm]{Re};
    \draw[line width=0.5pt, ->] (0,-1.7) -- (0,1.7) node[above,xshift = -0.25 cm, yshift = -0.4 cm]{Im};
    \filldraw[line width=1pt, color = black, fill = gray, fill opacity = 0.5] (0, 0) circle (1) node{};
\end{scope}
\end{tikzpicture} &
\begin{tikzpicture}[>=stealth]
\begin{scope}[scale = 0.6]
    \clip (-1.7,-1.7) rectangle (1.7,1.7);
    \draw[line width=0.5pt, ->] (-1.7,0) -- (1.7,0) node (X) [right,xshift = -0.5 cm, yshift = -0.2 cm]{Re};
    \draw[line width=0.5pt, ->] (0,-1.7) -- (0,1.7) node[above,xshift = -0.25 cm, yshift = -0.4 cm]{Im};
    \draw[line width=2pt, color=gray] (0,0) -- (1,0);
    \filldraw (0, 0) circle (2pt) node{};
    \filldraw (1, 0) circle (2pt) node{};
\end{scope}
\end{tikzpicture} &
\begin{tikzpicture}[>=stealth]
\begin{scope}[scale = 0.6]
    \clip (-1.7,-1.7) rectangle (1.7,1.7);
    \draw[line width=0.5pt, ->] (-1.7,0) -- (1.7,0) node (X) [right,xshift = -0.5 cm, yshift = -0.2 cm]{Re};
    \draw[line width=0.5pt, ->] (0,-1.7) -- (0,1.7) node[above,xshift = -0.25 cm, yshift = -0.4 cm]{Im};
    \node[regular polygon, regular polygon sides=8, shape border rotate=22.5, minimum size=1.2cm, fill=gray, line width=1pt, fill opacity=0.5] at (0,0) {};
    \draw[dashed] (0,0)  circle(1cm);
    
    \foreach \k in {0,...,7}{
        \filldraw ({cos(45*\k)}, {sin(45*\k)}) circle (2pt) node{};
    };
\end{scope}
\end{tikzpicture} &
\begin{tikzpicture}[>=stealth]
\begin{scope}[scale = 0.6]
    \clip (-1.7,-1.7) rectangle (1.7,1.7);
    \draw[line width=0.5pt, ->] (-1.7,0) -- (1.7,0) node (X) [right,xshift = -0.5 cm, yshift = -0.2 cm]{Re};
    \draw[line width=0.5pt, ->] (0,-1.7) -- (0,1.7) node[above,xshift = -0.25 cm, yshift = -0.4 cm]{Im};
    \foreach \i [count=\j from 0] in {0,40,80,120,160,200,240,280} {
        \coordinate (P\j) at ({cos(\i)},{sin(\i)});
        \filldraw (P\j) circle (2pt) node{};
    }
    \fill[fill=gray, fill opacity=0.5] 
    (P0) -- (P1) -- (P2) -- (P3) -- (P4) -- (P5) -- (P6) -- (P7) -- cycle;
    \draw[dashed] (0,0)  circle(1cm);
\end{scope}
\end{tikzpicture}
&
\begin{tikzpicture}[>=stealth]
\begin{scope}[scale = 0.6]
    \clip (-1.7,-1.7) rectangle (1.7,1.7);
    \draw[line width=0.5pt, ->] (-1.7,0) -- (1.7,0) node (X) [right,xshift = -0.5 cm, yshift = -0.2 cm]{Re};
    \draw[line width=0.5pt, ->] (0,-1.7) -- (0,1.7) node[above,xshift = -0.25 cm, yshift = -0.4 cm]{Im};
    \filldraw[line width=1pt, color = black, fill = gray, fill opacity = 0.5] (0, 0.5) circle (0.5) node{};
\end{scope}
\end{tikzpicture}
&
\begin{tikzpicture}[>=stealth]
\begin{scope}[scale = 0.6]
    \clip (-1.7,-1.7) rectangle (1.7,1.7);
    \draw[line width=0.5pt, ->] (-1.7,0) -- (1.7,0) node (X) [right,xshift = -0.5 cm, yshift = -0.2 cm]{Re};
    \draw[line width=0.5pt, ->] (0,-1.7) -- (0,1.7) node[above,xshift = -0.25 cm, yshift = -0.4 cm]{Im};
    \draw[dashed] (0, 0.5) circle (0.5) node{};
    \node[regular polygon, regular polygon sides=8, shape border rotate=22.5, minimum size=0.6cm, fill=gray, line width=1pt, fill opacity=0.5] at (0,0.5) {};
    \foreach \k in {0,...,7}{
        \filldraw ({0.5*cos(45*\k)}, {0.5*(sin(45*\k)+1)}) circle (2pt) node{};
    };
\end{scope}
\end{tikzpicture}
&
\begin{tikzpicture}[>=stealth]
\begin{scope}[scale = 0.6]
    \clip (-1.7,-1.7) rectangle (1.7,1.7);
    \draw[line width=0.5pt, ->] (-1.7,0) -- (1.7,0) node (X) [right,xshift = -0.5 cm, yshift = -0.2 cm]{Re};
    \draw[line width=0.5pt, ->] (0,-1.7) -- (0,1.7) node[above,xshift = -0.25 cm, yshift = -0.4 cm]{Im};
    
    \pgfmathsetmacro{\alphaVal}{2.0}
    \pgfmathsetmacro{\betaVal}{0.5}
    
    \filldraw[line width=1pt, color=black, fill=gray, fill opacity=0.5, samples=150, smooth, domain=0:360] 
        plot (\x : {(1-\betaVal) * ((1 + sin(\x))/2)^(\alphaVal) + \betaVal}) -- cycle;
\end{scope}
\end{tikzpicture}
\\
References & - & \cite{zhang2018on, arun2020rfocus, khaleel2022phase, nassirpour2023power}  & \cite{smith1983comparison, sohrabi2016hybrid, wang2018hybrid, wu2020beamforming} & \cite{cao2024ris, kutay2024received, xy2025effect} & \cite{shlezinger2019dynamic, rezvani2024channel, castellanos2025embracing} & \cite{di2020hybrid}  & \cite{abeywickrama2020intelligent, zhang2021performance}\tnote{*}
\\
\bottomrule
\end{tabular}
\begin{tablenotes}
\item[*] $\cW = \big\{r(\theta)e^{j\theta}:  r(\theta) = (1-\beta) \big(\frac{1+\sin \theta}{2}\big)^\alpha + \beta \big \}$, with the drawing corresponding to $\alpha = 2$ and $\beta=0.5$.
\end{tablenotes}
\end{threeparttable}
\end{table*}

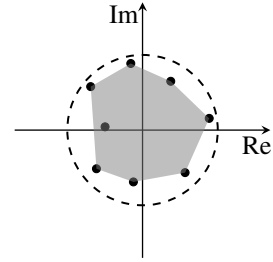
\begin{figure}
    \centering
    \begin{tikzpicture}[>=stealth]
    \begin{scope}
        \clip (-1.7,-1.7) rectangle (1.7,1.7);
        \draw[line width=0.5pt, ->] (-1.7,0) -- (1.7,0) node (X) [right,xshift = -0.5 cm, yshift = -0.2 cm]{Re};
        \draw[line width=0.5pt, ->] (0,-1.7) -- (0,1.7) node[above,xshift = -0.25 cm, yshift = -0.4 cm]{Im};
        \coordinate (P0) at (10:0.9);
        \coordinate (P1) at (60:0.75);
        \coordinate (P2) at (100:0.9);
        \coordinate (P3) at (140:0.9);
        \coordinate (P4) at (175:0.5);
        \coordinate (P5) at (220:0.8);
        \coordinate (P6) at (260:0.7);
        \coordinate (P7) at (315:0.8);
        \foreach \i in {0,1,...,7} {
            \filldraw (P\i) circle (1.5pt) node{};
        }
        \fill[fill=gray, fill opacity=0.5] 
        (P0) -- (P1) -- (P2) -- (P3) -- (P5) -- (P6) -- (P7) -- cycle;
        \draw[dashed, line width = 0.75pt] (0,0) circle (1cm);
    \end{scope}
    \end{tikzpicture}
    \caption{Response of $3$-bit phase shifter. 
    The dots are the elements of $\cW$ while the shaded region is their convex hull. The dashed circle is $\cW_{\text{ideal}}$. As the perimeter of the heptagonal convex hull equals 5.01, there exists some $\bw$ attaining $5.01/2\pi \approx 80\%$ of $g_{\text{ideal}}$, which corresponds to a $2$-dB shortfall in beamforming gain. Part of this shortfall is the power loss caused by the magnitude of the elements of $\cW$ being less than unity.
    }
    \label{fig:imperfection}
\end{figure}

The applicability of the above result is exemplified in Fig.~\ref{fig:imperfection}.
This applicability extends, besides point-to-point beamforming, to transmissions via a RIS. In that case, the role of the $N$ array antennas is played by the $N$ RIS elements and $h_n$ signifies the channel coefficient linking transmitter and receiver through the $n$th such element. Also included are settings where,
besides the transmitter-RIS-receiver connection, there is a direct transmitter-receiver path and the response coefficients at the RIS elements are drawn from
$\cW_M$: letting $h_0$ denote the direct path, the beamforming optimization
\begin{align}
    &\max_{\bw \in \cW_M^N} \bigg| h_0 + \sum_n w_n h_n \bigg| = \!\! \max_{(w_0, \bw) \in \cW_M^{N+1}} \bigg| w_0 h_0 + \sum_n w_n h_n \bigg| \label{direct_path}
\end{align}
is merely an augmented version of \eqref{original_problem} (see App. \ref{app:direct_path}).

The relevance of the result also extends to fading channels, and in particular to
$\{h_n\}$ drawn from any i.i.d. circularly symmetric distribution.
If $\{h_n\}$ are i.i.d. realizations of $h$, then, as $N$ grows large, 
the ideal beamforming gain (normalized by $N^2$ to keep it from diverging)
converges to $( \E[|h|] )^2$.
Relative to it, the worst-case shortfall with an arbitrary set of phase of phase shifts is again given by \eqref{constant_in_dB}. In fact, in this large-$N$ regime, the shortfall equals this worst-case value.



\section{Derivation of the Main Result}
\label{sec_derivation}

This section lays down the derivation of \eqref{main_result}--\eqref{best_constant}, beginning with a presentation of
the necessary mathematical machinery, chiefly the notion of support function.

\subsection{The Support Function}

Identifying $\bbC$ with $\bbR^2$, the inner product of two complex numbers can be defined as
\begin{align}
    \langle \cdot,\cdot \rangle : \bbC \times \bbC &\rightarrow \bbR \label{inner_product}\\
    (z,w) &\mapsto \Re{\overline{z}w}. \nonumber
\end{align}
The workhorse of the analysis that follows is the identity
\begin{align}
    |w| = \max_{\theta\in[0,2\pi]} \langle e^{j\theta}, w \rangle ,
    \label{workhorse}
\end{align}
which turns a nonlinear function into the optimization of a linear one, whereby
the beamforming optimization can be decoupled into
\begin{align}
     g_\cW & =  \max_{\bw \in \cW^N} \bigg| \sum_n w_n h_n \bigg|\\
    &=\max_{\bw \in \cW^N} \max_{\theta\in[0,2\pi]} \Big\langle e^{j\theta}, \sum_n w_n h_n \Big\rangle \\
    &= \max_{\bw \in \cW^N} \max_{\theta\in [0,2\pi]}  \sum_n \big\langle e^{j\theta}, w_nh_n \big\rangle\\
    &= \max_{\theta\in [0,2\pi]}    \sum_n \max_{w \in \cW} \big\langle e^{j\theta}, wh_n \big\rangle\\
    &= \max_{\theta\in [0,2\pi]}    \sum_n |h_n| \max_{w \in \cW} \big\langle e^{j(\theta-\theta_n)}, w \big\rangle.
    \label{one_dimensional_search}
\end{align}

The inner optimization problem in \eqref{one_dimensional_search} is of the form $\max\limits_{w\in S} \langle z,w \rangle$, hence the mapping 
\begin{align}
    \bbC&\rightarrow \bbR \label{support_function}\\
    z&\mapsto \max_{w\in S} \langle z,w \rangle \nonumber
\end{align}
becomes of interest. This mapping, well defined because $S$ is a compact set, depends only on the convex hull of $S$.

\begin{lemma}
\label{lemma:convex_hull_and_support_function}
For a compact set $S\subset \bbC$,
\begin{align}
    \max_{w\in S}\langle z,w \rangle = \max_{w\in \conv\,S}\langle z,w \rangle.
\end{align}
\end{lemma}
\begin{proof}
This is a special case of the result for convex functions in \cite[Thm. 32.2]{rockafellar1997convex}.
\end{proof}

The mapping in \eqref{support_function} can be identified with the \textit{support function} of $\conv\,S$ \cite[Ch. 1.7]{schneider2013convex}. A useful property of convex hull is provided next.
 
\begin{lemma}
\label{lemma:convex_hull_and_minkowski_sum}
For sets $S_1$ and $S_2$ in $\bbC$,
\begin{align}
    \conv(S_1 + S_2) = \conv\, S_1 + \conv\, S_2.
\end{align}
\end{lemma}
\begin{proof}
See \cite[Thm. 1.1.2]{schneider2013convex}.
\end{proof}

\subsection{A Lower Bound}

Let us next bound $g_\cW$ on the basis that the maximum is at least the average and that the mapping
\begin{align}
    \theta \mapsto \max\limits_{w \in \cW}\big\langle e^{j\theta}, w \big\rangle \label{support_function_on_unit_circle}
\end{align}
is periodic, with period $2\pi$.


\begin{theorem}
\label{thm:main_result}
The inequality in \eqref{main_result} holds with
\begin{align}
 \text{constant} =  \frac{1}{2\pi}\int_0^{2\pi} \max_{w \in \cW} \big\langle e^{j\theta}, w \big\rangle d\theta. \label{perimeter_in_disguise}
\end{align}
\end{theorem}
\begin{proof}
See App. \ref{Ferretti}.
\end{proof}

In the special case that $\cW = \{0,1\}$, it can be verified that 
\begin{align}
    \max_{w \in \cW}\big\langle e^{j\theta}, w \big\rangle = \max(\cos\theta, 0),
\end{align}
whose integral over $[0,2\pi]$ equals $2$. This turns Thm. \ref{thm:main_result} into
\begin{align}
  g_\cW \geq   \frac{1}{\pi} \sum_n |h_n|.
\end{align}
In the context of RIS-assisted communication, a similar result is available in \cite[Thm.~2]{arun2020rfocus}, but only for $N \to \infty$ and under certain assumptions on the
values that the channel can take.
This special case, $\cW = \{0,1\}$,
is further presented in \cite{kaufman1966inequality, bledsoe1970inequality, daykin1974sets}, in the slightly different (but ultimately equivalent) form
of showing
%
that there exists a subset $S\subset\{1,\ldots,N\}$ satisfying
\begin{align}
    \bigg|\sum_{n\in S} h_n \bigg| \geq \frac{1}{\pi} \sum_n |h_n|. \label{binary_case}
\end{align}
This formulation appears in various textbooks as an exercise \cite[Ch. VIII, Exercise 3.1]{bourbaki1966general}, \cite[Exercise 14.9]{steele2004cauchy}. A weaker version with a constant $\frac{1}{6}$ appeared in the 1986 Chinese Mathematical Olympiad \cite[pp. 23]{liu1998chinese}.

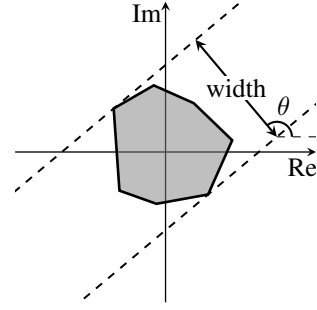
\begin{figure}
    \centering
    \begin{tikzpicture}[>=stealth]
    \begin{scope}
        \clip (-2,-2) rectangle (2,2);
        \draw[line width=0.5pt, ->] (-2,0) -- (2,0) node (X) [right,xshift = -0.5 cm, yshift = -0.2 cm]{Re};
        \draw[line width=0.5pt, ->] (0,-2) -- (0,2) node[above,xshift = -0.25 cm, yshift = -0.4 cm]{Im};
        \coordinate (P0) at (10:0.9);
        \coordinate (P1) at (60:0.75);
        \coordinate (P2) at (100:0.9);
        \coordinate (P3) at (140:0.9);
        \coordinate (P4) at (175:0.5);
        \coordinate (P5) at (220:0.8);
        \coordinate (P6) at (260:0.7);
        \coordinate (P7) at (315:0.8);
        \filldraw[fill=gray, fill opacity=0.5, line width=1pt] 
        (P0) -- (P1) -- (P2) -- (P3) -- (P5) -- (P6) -- (P7) -- cycle;

        \def\alpha{40}
        \draw[dashed, line width = 0.75pt]($(P7) - 3*({cos(\alpha)},{sin(\alpha)})$)--($(P7) + 3*({cos(\alpha)},{sin(\alpha)})$);
        \draw[dashed, line width = 0.75pt]($(P3) - 3*({cos(\alpha)},{sin(\alpha)})$)--($(P3) + 3*({cos(\alpha)},{sin(\alpha)})$);

        \coordinate (Start) at ($(P3) + (\alpha:1.42)$);
        \coordinate (End) at ($(P7)!(Start)!($(P7)+(\alpha:1)$)$);
        \draw[<->, thick] (Start) -- (End) node[midway, fill=white, inner sep=1pt] {width};
        \draw[dashed] (End) -- ($(End) + (1,0)$);
        \draw[thick] ($(End) + (0.2,0)$) arc (0:130:0.2) node[above right]{$\theta$};
    \end{scope}
    \end{tikzpicture}
    \caption{Width of a convex set in the direction specified by $\theta$.}
    \label{fig:width}
\end{figure}

\subsection{Geometric Interpretation}

Results from convex geometry allow interpreting the constant in \eqref{perimeter_in_disguise} geometrically. The width of a compact convex set $S\subset \bbC$ in the direction of $\theta$ is defined as \cite[Sec. 1.7]{schneider2013convex}
\begin{align}
    &\max_{w\in\cW}\big\langle e^{j\theta}, w \big\rangle + \max_{w\in\cW}\big\langle e^{j(\theta+\pi)}, w \big\rangle
\end{align}
and illustrated in Fig. \ref{fig:width}.
The mean width of $S$ is then
\begin{align}
    & \!\! \frac{1}{2\pi}\int_0^{2\pi} \left( \max_{w\in\cW}\big\langle e^{j\theta}, w \big\rangle + \max_{w\in\cW}\big\langle e^{j(\theta+\pi)}, w \big\rangle \right) d\theta  \\
    &\qquad\qquad\qquad\qquad\qquad\qquad  = \frac{1}{\pi}\int_0^{2\pi} \max_{w\in\cW}\big\langle e^{j\theta}, w \big\rangle\, d\theta
    \nonumber
\end{align}
from the periodicity of \eqref{support_function_on_unit_circle}. One can relate this mean width with the perimeter of $S$ using the Cauchy's surface area formula \cite[Sec. 5.3]{schneider2013convex}, whose two-dimensional instance is
presented next as Thm. \ref{thm:cauchy}; see \cite{vouk1948projected, meltzer1949shadow} for intuitive explanations. Application of the formula directly gives, as desired,
\begin{align}
    \frac{1}{2\pi}\int_0^{2\pi} \max_{w\in\cW}\big\langle e^{j\theta}, w \big\rangle d\theta =  \frac{\text{perimeter of }\conv\,\cW}{2\pi} . \label{the_constant}
\end{align}

\begin{theorem} \textit{(Cauchy's surface area formula)}
\label{thm:cauchy}
For a compact convex set $S\subset\bbC$,
\begin{align}
    \text{mean width of $S$} = \frac{\text{perimeter of $S$}}{\pi}.
\end{align}
\end{theorem}

\subsection{Comparison with an Existing Result}

For the special case of $\cW=\cW_M$, the bound in \eqref{main_result} is available in
\cite[Prop. 2]{han2019large}, albeit with the looser constant
\begin{align}
    \min_{\theta \in [0,2\pi]} \max_{w\in\cW_M}\big\langle e^{j\theta}, w \big\rangle = \cos\frac{\pi}{M}.\label{crude_constant}
\end{align}
This result readily follows from \eqref{one_dimensional_search} via
\begin{align}
    g_\cW 
    & = \max_{\theta\in [0,2\pi]}    \sum_n |h_n| \, \max_{w\in\cW_M}\big\langle e^{j(\theta-\theta_n)}, w \big\rangle  \\
    &\geq \min_{\theta\in [0,2\pi]}  \sum_n |h_n| \, \max_{w\in\cW_M}\big\langle e^{j(\theta-\theta_n)}, w \big\rangle  \\
    &\geq  \sum_n |h_n| \, \min_{\theta\in [0,2\pi]} \max_{w\in\cW_M}\big\langle e^{j(\theta-\theta_n)}, w \big\rangle  \\
    & =  \min_{\theta \in [0,2\pi]} \max_{w\in\cW_M}\big\langle e^{j\theta}, w \big\rangle \cdot \sum_n |h_n| \\
    & = \cos \frac{\pi}{M} \cdot g_\text{ideal}    .
    \label{prior_art_bound}
\end{align}
The constant in \eqref{crude_constant} is indeed looser than \eqref{best_constant} because the average is at least the minimum, that is,
\begin{align}
    \frac{1}{2\pi}\int_0^{2\pi}  \max_{w\in\cW_M}\big\langle e^{j\theta}, w \big\rangle \,d\theta \geq \min_{\theta \in [0,2\pi]}  \max_{w\in\cW_M}\big\langle e^{j\theta}, w \big\rangle.
\end{align}
As \eqref{crude_constant} is the radius of the circle inscribed in $\conv\,\cW_M$, the inequality corresponds (see Fig. \ref{fig:constant_comparison}) to
\begin{align}
    &2\pi\cdot\text{(radius of the inscribed circle)}  \leq \text{perimeter of $\conv\,\cW$}.\nonumber
\end{align}

\begin{figure}
    \centering
    \begin{tikzpicture}[>=stealth]
    \begin{scope}[]
        \clip (-1.7,-1.7) rectangle (1.7,1.7);
        
        \node[regular polygon, regular polygon sides=6, shape border rotate=0, minimum size=2cm, draw=black, fill=black, line width=1pt, fill opacity = 0.6] at (0,0) {};
        \filldraw[line width=1pt, draw=black, fill=white, fill opacity=0.8] (0,0)  circle (0.8660cm);
        \draw[dashed, line width=0.75pt, color=black] (0,0)  circle (1cm);

        \draw[line width=0.5pt, ->] (-1.7,0) -- (1.7,0) node (X) [right,xshift = -0.5 cm, yshift = -0.2 cm]{Re};
        \draw[line width=0.5pt, ->] (0,-1.7) -- (0,1.7) node[above,xshift = -0.25 cm, yshift = -0.4 cm]{Im};
        
        \foreach \k in {0,...,5}{
            \filldraw ({cos(60*\k)}, {sin(60*\k)}) circle (1.5pt) node{};
        };
    \end{scope}
    \end{tikzpicture}
    \caption{Visual comparison of the constants in \eqref{crude_constant} and \eqref{best_constant}, which respectively correspond to the perimeter of the circle in light gray and to
    the perimeter of the regular $M$-gon in charcoal gray, both  divided by $2\pi$. As a reference, the unit circle is shown as a dashed line.}
    \label{fig:constant_comparison}
\end{figure}
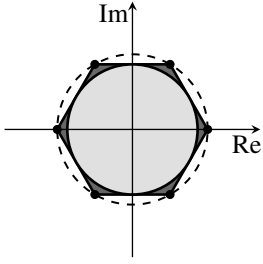

The difference between the constants in \eqref{crude_constant} and \eqref{best_constant} can be
conveniently quantified asymptotically. For large $M$, \eqref{crude_constant} expands as
\begin{align}
    1 - \frac{\pi^2}{2} \cdot \frac{1}{M^2} + O\bigg(\frac{1}{M^4}\bigg)
\end{align}
whereas, for a regular $M$-gon, \eqref{best_constant} becomes\footnote{The expansion in \eqref{asymptotic_expression} coincides with that of a phase quantization error in the array processing literature \cite[Eqn. 75]{mailloux1982phased}.}
\begin{align}
    \frac{1}{2\pi}\cdot 2M\sin\frac{\pi}{M}
    & =1 - \frac{\pi^2}{6}\cdot \frac{1}{M^2} + O\bigg(\frac{1}{M^4}\bigg). \label{asymptotic_expression}
\end{align}
The latter approaches unity three times faster.



\subsection{Tightness}

As the next result formalizes, for Thm. \ref{thm:main_result} to hold for any $N$ and $h_1, \ldots, h_N$, the constant in \eqref{perimeter_in_disguise} cannot be further tightened.
\begin{theorem}
\label{thm:optimality_of_constant}
If \eqref{main_result} holds for any $N$ and $h_1, \ldots, h_N$, then
\begin{align}
    \text{constant} \leq  \frac{1}{2\pi}\int_0^{2\pi} \max\limits_{w \in \cW}\big\langle e^{j\theta}, w \big\rangle  d \theta.
\end{align}
\end{theorem}
\begin{proof}
See App. \ref{rain}.
\end{proof}





Although Thm. \ref{thm:optimality_of_constant} affirms that the constant in \eqref{perimeter_in_disguise} cannot be further tightened, it is, at its heart, a worst-case analysis. (In some cases, e.g., $\cW$ intersecting with the unit circle and $h_1 = \cdots = h_n$, the ideal beamforming gain can even be attained.) Transcending the notion of worst case to embrace other notions, say the average beamforming gain, requires accepting that the channel is an instance of a fading distribution. 

\subsection{Fading Channels}

Let us posit that $\{h_n\}$ are drawn from an i.i.d. circularly symmetric distribution, say they are i.i.d. realizations of $h$ with $\bbE\big[|h|\big]<\infty$. As $N$ grows large, light can be shed on the beamforming gain that goes beyond its worst-case value.
%
Specifically, as formalized next, the bound established in Thm.~\ref{thm:main_result} 
(with both sides normalized by $N$ to keep it from diverging)
becomes tight, in the sense that both sides converge to the same nonrandom limit
\begin{align}
    \bbE\big[|h|\big] \cdot \frac{1}{2\pi}\int_0^{2\pi} \max\limits_{w \in \cW}\big\langle e^{j\theta}, w \big\rangle  d \theta. \label{constant}
\end{align}

\begin{theorem}
For $N \to \infty$, the sequence of random variables $\frac{g_{\cW}}{N}$ converges almost surely (a.s.) to \eqref{constant}.
\end{theorem}
\begin{proof}
See App. \ref{app:law_of_large_numbers}.
\end{proof}

The expression in \eqref{constant_in_dB} therefore represents, for large $N$, not only the worst-case shortfall, but actually the shortfall attained a.s. for every channel realization. And, as formalized next, it also represents the average shortfall under the sole proviso that, as dictated by physics, $\E[|h|^2] < \infty$.

\begin{theorem}
If $\bbE\big[|h|^p\big]<\infty$ for some $p\geq 1$, the sequence of random variables $\frac{g_{\cW}}{N}$ converges in mean of order $p$ to \eqref{constant},
\begin{align}
    \frac{\bbE \big[g^p_{\cW} \big]}{N^p} \rightarrow \big(\bbE\big[|h|\big]\big)^p \cdot \left(\frac{1}{2\pi}\int_0^{2\pi} \max\limits_{w \in \cW}\big\langle e^{j\theta}, w \big\rangle  d \theta\right)^{\!p}. \label{second_part}
\end{align}
\end{theorem}
\begin{proof}
See App. \ref{app:uniform_integrability}.
\end{proof}

For $p=2$, the above result implies that the shortfall in average beamforming gain abides by \eqref{constant_in_dB}.


\section{A Geometric Proof}
\label{geometric_proof}

This section provides an alternative perspective on the problem.
Following \cite{do2026optimum}, one can recast
\begin{align}
    \max_{\bw \in \cW^N} \bigg| \sum_n w_n h_n \bigg|
\end{align}
as
\begin{align}
    \max_{w} & \quad |w|\\
    \text{s.t.} & \quad w\in h_1\cW + \cdots + h_N \cW, \nonumber
\end{align}
where the summations are Minkowski sums of sets. By means of Lemma \ref{lemma:convex_hull_and_support_function_absolute} presented below, this can be reformulated as
\begin{align}
    \max_{w} & \quad |w| \label{reformulation}\\
    \text{s.t.} & \quad w\in \conv(h_1\cW + \cdots + h_N \cW). \nonumber
\end{align}

\begin{lemma}
\label{lemma:convex_hull_and_support_function_absolute}
For any compact set $S\subset \bbC$,
\begin{align}
    \max_{w\in S}|w| = \max_{w\in \conv\,S}|w|.
\end{align}
\end{lemma}
\begin{proof}
See App. \ref{loto}.
\end{proof}



In turn, the Cauchy's surface area formula paves the way for two other results of interest. 
\begin{corollary}
\label{cor:perimeter_sum}
For compact convex sets $S_1$ and $S_2$ in $\bbC$,
\begin{align}
    \text{perimeter of $S_1+S_2$} = \text{perimeter of $S_1$} + \text{perimeter of $S_2$}. \nonumber
\end{align}
\end{corollary}
\begin{proof}
See App. \ref{nadal}.
\end{proof}

\begin{corollary}
\label{cor:perimeter_inclusion}
Given two compact convex sets $S_1$ and $S_2$, if $S_1 \subset S_2 \subset \bbC$, the perimeter of $S_1$ is at most that of $S_2$.
\end{corollary}
\begin{proof}
Since, on any direction, the width of $S_1$ is at most that of $S_2$, the result follows from the Cauchy's surface area formula.
\end{proof}

Armed with these corollaries, the main result in \eqref{main_result}--\eqref{best_constant} can be proved in an alternative fashion.
As per Lemma~\ref{lemma:convex_hull_and_minkowski_sum},
\begin{align}
    &\conv(h_1\cW + \cdots + h_N \cW)   \nonumber \\
    & \qquad = h_1 \conv\, \cW + \cdots + h_N \conv\, \cW \label{minkowski_sum}
\end{align}
and, applying repeatedly Cor. \ref{cor:perimeter_sum}, we have that
\begin{align}
    &\text{perimeter of } h_1 \conv\, \cW + \cdots + h_N \conv\, \cW \\
    &=\text{perimeter of } h_1 \conv\, \cW + \cdots + \text{perimeter of } h_N \conv\, \cW \nonumber
    \\
    &= \sum_n \text{perimeter of } \conv\, \cW \cdot |h_n| 
    \\
    &= \text{perimeter of } \conv\, \cW \cdot \sum_n|h_n| 
    \\
    &=\text{perimeter of }\conv\,\cW \cdot g_{\text{ideal}}. \label{perimeter_of_minkowski_sum}
\end{align}
As the optimal objective of \eqref{reformulation} is $g_{\text{ideal}}$,
\eqref{minkowski_sum} is contained in the disk
$    \{z\in\bbC: |z|\leq g_\cW \}$.
Applying Cor. \ref{cor:perimeter_inclusion}, one obtains
\begin{align}
    2\pi g_\cW
    &= \text{perimeter of } \{z\in\bbC: |z|\leq g_\cW \}\\
    &\geq \text{perimeter of } h_1 \conv\, \cW + \cdots + h_N \conv\, \cW\\
    &=\text{perimeter of }\conv\,\cW \cdot g_{\text{ideal}}.
\end{align}

This geometric proof also provides intuition on why Thm.~\ref{thm:optimality_of_constant} must hold. From the reformulation in \eqref{reformulation}, 
the set \eqref{minkowski_sum} has an element whose modulus is $g_{\cW}$. For $ h_n = e^{j\frac{2\pi n}{N}}$, this set,
\begin{align}
      e^{j\frac{2\pi\cdot 1}{N}}  \conv\, \cW + e^{j\frac{2\pi\cdot 2}{N}}  \conv\, \cW+ \cdots + e^{j\frac{2\pi\cdot N}{N}} \conv\, \cW, \nonumber 
\end{align}
is invariant under $\frac{2\pi}{N}$-rotation.
From this rotational invariance and the convexity, \eqref{minkowski_sum} includes a regular $N$-gon inscribed in a disk of radius $g_{\cW}$.
By means of Cor.~\ref{cor:perimeter_inclusion}, \eqref{perimeter_of_minkowski_sum} is bounded below by the perimeter of $\conv\,\cW_N \cdot g_{\cW}$,
meaning that
\begin{align}
    &\text{perimeter of }\conv\,\cW \cdot g_{\text{ideal}}\geq \text{perimeter of }\conv\,\cW_N \cdot g_{\cW}. \nonumber
\end{align}
Letting $N\rightarrow\infty$, the above becomes
\begin{align}
    g_{\cW} \leq \frac{\text{perimeter of }\conv\,\cW}{2\pi} \cdot g_{\text{ideal}}
\end{align}
for $h_n = e^{j\frac{2\pi n}{N}}$, 
which reaffirms Thm. \ref{thm:optimality_of_constant}.

A special case of the foregoing geometric proof, corresponding to $\cW=\{0,1\}$ as per \eqref{binary_case},
is given in \cite[Thm.~3]{daykin1974sets}, in a derivation attributed to S. Kakutani.

\section{A Minute Tightening}
\label{sec_tightening}

According to Thm. \ref{thm:optimality_of_constant}, the obtained constant is optimal in the sense of the bound holding for any $N$. 
For a specific $N$, though, there is room for some slight further tightening when $\conv \,\cW$ is a polygon. Indeed, for each $N$ one can determine the constant satisfying \eqref{main_result}
for any $h_1,\ldots, h_N$ and, as $N$ grows, this constant decreases monotonically and converges to \eqref{best_constant}. A bound on this constant for specific $N$
is given next. 

\begin{theorem}
\label{thm:refinement}
Provided that $\conv \,\cW$ is a polygon with $M$ vertices, \eqref{main_result} holds with the constant
\begin{align}
    \frac{\text{perimeter of }\conv\,\cW}{\text{perimeter of }\conv\,\cW_{MN}} = \frac{\text{perimeter of }\conv\,\cW}{2MN \sin \frac{\pi}{MN}}. \label{constant_fixed_N}
\end{align}
\end{theorem}
\begin{proof}
See App. \ref{mac}
\end{proof}

The maximum discrepancy between \eqref{best_constant} and \eqref{constant_fixed_N}, pertaining to the case $M=N=2$, amounts to $0.9$-dB difference in beamforming gain.

For the special case of a regular $M$-gon, $\cW = \cW_M$, the constant in \eqref{constant_fixed_N} cannot be improved upon and it reduces to
\begin{align}
    \frac{\sin \frac{\pi}{M}}{N \sin \frac{\pi}{M N}} ,
\end{align}
which generalizes the result proved in \cite[Thm. 3.1]{grundbacher2025sharp} for $\cW_2$.
As $N$ grows large, the above reverts as anticipated to \eqref{best_constant}, which, for a regular $M$-gon, adopts the form
\begin{align}
    \frac{M}{\pi} \sin \frac{\pi}{M} .
\end{align}
This formula returns shortfalls of $-3.9$ dB, $-0.9$ dB, and $-0.2$ dB, respectively for $1$-bit, $2$-bit, and $3$-bit discrete uniform phase shifters.

Beyond the realm of polygons, the best constant can also be determined for $\cW = \{0,1\}$ in some cases. It is
$\frac{1}{2}$ for $N=2$, and $\frac{1}{3}$ for $N=3$ obtained from a trivial bound
\begin{align}
    \max_{\bw \in \{0,1\}^N} \bigg| \sum_n w_n h_n \bigg| \geq \max_n |h_n| \geq \frac{1}{N} \sum_n |h_n|.
\end{align} 

\section{Conclusion}
\label{sec_conclusion}

Given any arbitrary set of feasible phase shifts, \eqref{constant_in_dB}
quantifies the maximum penalty in transmit or receive beamforming gain, or in RIS-assisted beamforming gain, relative to ideality.
In i.i.d. fading channels, the penalty concentrates to its maximum as the number of antennas grows.
As potential follow-up work, one could entertain the generalization to correlated fading channels,
in particular the assessment of whether \eqref{constant_in_dB} continues to hold.

\section{Acknowledgment}
Heedong Do would like to recognize \cite[Exercise 14.9]{steele2004cauchy}, which brought \cite{bledsoe1970inequality} to his attention.

\appendices

\section{}
\label{app:direct_path}


For any $h_0, \ldots, h_N\in\bbC$,
\begin{align}
    & \!\!\!\!\!\!\!\!\!\!\!\!\!\!\!\! \max_{(w_0, \bw) \in \cW_M^{N+1}} \bigg| w_0 h_0 + \sum_n w_n h_n \bigg| \nonumber \\
    &\qquad\qquad = \!\! \max_{(w_0, \bw) \in \cW_M^{N+1}} \bigg| h_0 + \sum_n \frac{w_n}{w_0} h_n \bigg| \\
    &\qquad\qquad \leq \max_{\bw \in \cW_M^N} \bigg| h_0 + \sum_n w_n h_n \bigg|,
\end{align}
where the inequality holds because $\big\{\frac{w_n}{w_0}: w_0, w_n \in \cW \big\} = \cW$.
The inequality in the reverse direction trivially holds.



\section{}
\label{Ferretti}
Recalling \eqref{one_dimensional_search},
\begin{align}
    \max_{\bw \in \cW^N} \bigg| \sum_n w_n h_n \bigg| = \max_{\theta\in [0,2\pi]} \sum_n |h_n| \max_{w \in \cW} \big\langle e^{j(\theta-\theta_n)}, w \big\rangle.
\end{align}
As the maximum is at least the average, a lower bound is given by
\begin{align}
    &\frac{1}{2\pi}\int_0^{2\pi} \sum_n |h_n| \max_{w \in \cW} \big\langle e^{j(\theta-\theta_n)}, w \big\rangle d \theta \nonumber \\
    &\qquad\qquad =  \sum_n |h_n| \cdot \frac{1}{2\pi}\int_0^{2\pi}\max_{w \in \cW}\big\langle e^{j(\theta-\theta_n)}, w \big\rangle d \theta. \label{maximum_is_at_least_average}
\end{align}
From the periodicity of \eqref{support_function_on_unit_circle}, it holds that
\begin{align}
    \frac{1}{2\pi}\int_0^{2\pi}\max_{w \in \cW}\big\langle e^{j(\theta-\theta_n)}, w \big\rangle d \theta = \frac{1}{2\pi}\int_0^{2\pi}\max_{w \in \cW}\big\langle e^{j\theta}, w \big\rangle d \theta
\end{align}
for all $n$.
This turns \eqref{maximum_is_at_least_average} into
\begin{align}
     \frac{1}{2\pi}\int_0^{2\pi}\max_{w \in \cW}\big\langle e^{j\theta}, w \big\rangle d \theta \cdot \sum_n |h_n|,
\end{align}
as desired.

\section{}
\label{rain}

It suffices to find one concrete such $N$ and $h_1,\ldots,h_N$. For a given $N$, let
\begin{align}
  h_n = e^{j\frac{2\pi n}{N}} . \label{venezuela}
\end{align}
The strategy shall be letting $N\rightarrow \infty$ on both sides of
\begin{align}
    \frac{1}{N}\max_{\bw \in \cW^N} \bigg| \sum_n w_n h_n \bigg| \geq \text{constant} \cdot \frac{1}{N}\sum_n |h_n|. \label{reis}
\end{align}
From \eqref{one_dimensional_search}, the left-hand side equals
\begin{align}
    \max_{\theta\in [0,2\pi]} \frac{1}{N}\sum_n \max_{w\in\cW}\big\langle e^{j(\theta-\frac{2\pi n}{N})}, w \big\rangle. \label{riemann_sum}
\end{align}
As \eqref{support_function_on_unit_circle} is a continuous function over a compact set, it is uniformly continuous \cite[Thm. 4.42]{pugh2015real}.
Therefore, the convergence
\begin{align}
    \frac{1}{N} \sum_n  \max_{w\in\cW}\big\langle e^{j(\theta-\frac{2\pi n}{N})}, w \big\rangle \rightarrow \frac{1}{2\pi}\int_0^{2\pi} \max_{w\in\cW}\big\langle e^{j\theta}, w \big\rangle \, d\theta \nonumber
\end{align}
is uniform \cite[Exercise 4.51]{pugh2015real}. For $N\rightarrow \infty$, the left-hand side of \eqref{reis} becomes
\begin{align}
    &\lim_{N\rightarrow \infty}\max_{\theta\in [0,2\pi]} \frac{1}{N}\sum_n  \max_{w\in\cW}\big\langle e^{j(\theta-\frac{2\pi n}{N})}, w \big\rangle \nonumber \\
    &\qquad\quad= \max_{\theta\in [0,2\pi]} \lim_{N\rightarrow \infty}\frac{1}{N}\sum_n  \max_{w\in\cW}\big\langle e^{j(\theta-\frac{2\pi n}{N})}, w \big\rangle\\
    &\qquad\quad= \max_{\theta\in [0,2\pi]} \frac{1}{2\pi}\int_0^{2\pi}  \max_{w\in\cW} \big\langle e^{j\theta}, w \big\rangle \, d\theta\\
    &\qquad\quad= \frac{1}{2\pi}\int_0^{2\pi}  \max_{w\in\cW} \big\langle e^{j\theta}, w \big\rangle \, d\theta,
\end{align}
where the limit and the maximum can be interchanged thanks to the uniform convergence.
As of the right-hand side of \eqref{reis}, it directly equals the constant because, in light of \eqref{venezuela}, $\frac{1}{N}\sum_n |h_n| = 1$.
Altogether,
\begin{align}
    \frac{1}{2\pi}\int_0^{2\pi} \max\limits_{w \in \cW}\big\langle e^{j\theta}, w \big\rangle  \, d\theta \geq \text{constant}
\end{align}
as desired.

\section{}
\label{app:law_of_large_numbers}
Recalling \eqref{one_dimensional_search}, it suffices to show that
\begin{align}
    \frac{g_\cW}{N}&=\max_{\theta\in [0,2\pi]} \frac{1}{N}\sum_n |h_n| \max_{w \in \cW} \big\langle e^{j(\theta-\theta_n)}, w \big\rangle \\
    & \stackrel{\text{a.s.}}{\to} \bbE\big[|h|\big] \cdot \frac{1}{2\pi}\int_0^{2\pi} \max_{w\in\cW}\big\langle e^{j\theta} , w \big\rangle d\theta . \label{convergence_to_constant}
\end{align}
To that purpose, let us render explicit the sample space $\Omega$, on which the random variables $h_1, h_2, \ldots$ and $\theta_1, \theta_2, \ldots$ live, and its probability function $P$.
The goal here is to find $\Omega'\subset \Omega$ satisfying $P(\Omega')=1$ and
\begin{align}
    &\max_{\theta\in [0,2\pi]} \frac{1}{N}\sum_n |h_n(\omega)| \max_{w \in \cW} \big\langle e^{j(\theta-\theta_n(\omega))}, w \big\rangle \label{uniform_convergence}\\
    &\qquad\qquad\qquad\qquad \rightarrow \bbE\big[|h|\big] \cdot \frac{1}{2\pi}\int_0^{2\pi} \max_{w\in\cW}\big\langle e^{j\theta}, w \big\rangle d\theta \nonumber
\end{align}
for any $\omega\in \Omega'$.

From the strong law of large numbers \cite[Cor. 5.2.1]{chow1997probability}, for every $\theta\in[0,2\pi]$ there exists $\Omega_\theta\subset \Omega$ satisfying $P(\Omega_\theta)=1$ and
\begin{align}
    &\frac{1}{N}\sum_n |h_n(\omega)| \max_{w \in \cW} \big\langle e^{j(\theta-\theta_n(\omega))}, w \big\rangle \label{first_slln}\\
    &\qquad\qquad\qquad \rightarrow \bbE\big[|h|\big] \cdot \frac{1}{2\pi}\int_0^{2\pi} \max_{w\in\cW}\big\langle e^{j\theta}, w \big\rangle d\theta \nonumber
\end{align}
for any $\omega\in \Omega_\theta$. Repeating the argument, we also have $A \subset \Omega$ such that $P(A)=1$ and
\begin{align}
    &\frac{1}{N}\sum_n |h_n(\omega)| \rightarrow \bbE\big[|h|\big]. \label{second_slln}
\end{align}
As a countable intersection of events of probability one is itself of probability one \cite[Exercise 1.5.5]{chow1997probability}, $P(\Omega')=1$ with
\begin{align}
    \Omega' = A \cap \bigcap_{\theta\in[0,2\pi]\cap \bbQ} \Omega_\theta.
\end{align}
For $\omega\in \Omega'$ fixed, one can think of the left-hand side of \eqref{first_slln},
\begin{align}
    \frac{1}{N}\sum_n |h_n(\omega)| \max_{w \in \cW} \big\langle e^{j(\theta-\theta_n(\omega))}, w \big\rangle, \label{sequence_of_functions}
\end{align}
as a sequence of functions of $\theta$. Two properties of \eqref{sequence_of_functions} come in handy:
\begin{enumerate}
    \item By construction, for $\theta\in[0,2\pi]\cap \bbQ$, \eqref{sequence_of_functions} converges pointwisely to the nonrandom limit,
    \begin{align}
        \bbE\big[|h|\big] \cdot \frac{1}{2\pi}\int_0^{2\pi} \max_{w\in\cW}\big\langle e^{j\theta}, w \big\rangle d\theta.
    \end{align}
    \item Each summand in \eqref{sequence_of_functions} is Lipschitz continuous as 
    \begin{align}
    &\Big|\max_{w\in\cW}\big\langle e^{j\theta}, w \big\rangle - \max_{w\in\cW}\big\langle e^{j\theta'}, w \big\rangle \Big|\\
    &\qquad\qquad \leq \Big|\max_{w\in\cW}\Big(\big\langle e^{j\theta}, w \big\rangle - \big\langle e^{j\theta}, w \big\rangle\Big) \Big|\\
    &\qquad\qquad = \Big|\max_{w\in\cW}\big\langle e^{j\theta}-e^{j\theta'}, w \big\rangle\Big|\\
    &\qquad\qquad \leq \big|e^{j\theta}-e^{j\theta'}\big| \max_{w\in\cW}|w|\\
    &\qquad\qquad \leq |\theta - \theta'| \max_{w\in\cW}|w|.
    \end{align}
    Therefore, \eqref{sequence_of_functions} is a function with Lipschitz constant
    \begin{align}
        \max_{w\in\cW}|w| \cdot \frac{1}{N}\sum_n |h_n(\omega)|. \label{lipschitz_constant}
    \end{align}
    From \eqref{second_slln} and the fact that a convergent sequence is bounded, \eqref{lipschitz_constant} is bounded above. This means the sequence of functions in \eqref{sequence_of_functions} is equicontinuous \cite[Exercise 4.15]{pugh2015real}.
\end{enumerate}
These two observations allow for the application of the Arzel\`a-Ascoli propagation theorem \cite[Thm. 4.16]{pugh2015real}, extending the pointwise convergence over $[0,2\pi]\cap \bbQ$ to the uniform convergence over $[0,2\pi]$. This proves \eqref{convergence_to_constant} and, with that, the convergence of $\frac{g_\cW}{N}$ to \eqref{constant}.

\section{}
\label{app:uniform_integrability}


From \cite[Example 4.3.1]{chow1997probability}, given that $\big\{|h_n|\big\}$ are i.i.d. random variables with $\bbE\big[|h|^p\big]<\infty$,
\begin{align}
    \bigg(\frac{g_{\text{ideal}}}{N}\bigg)^{\!p} = \bigg(\frac{1}{N}\sum_n |h_n|\bigg)^{\!p}
\end{align}
is uniformly integrable. As
\begin{align}
     0 \leq \bigg(\frac{g_\cW}{N}\bigg)^{\!p} \leq \bigg(\frac{g_{\text{ideal}}}{N}\bigg)^{\!p},   
\end{align}
$\big(\frac{g_\cW}{N}\big)^p$
is uniformly integrable as well \cite[pp. 94]{chow1997probability}. In conjunction with the convergence in probability, implied by the a.s. convergence \cite[Cor. 3.3.1]{chow1997probability}, $\big(\frac{g_\cW}{N}\big)^p$ converges in mean of order $p$ to the same limit \cite[Thm. 4.2.3]{chow1997probability}. 

\section{}
\label{loto}
Although the proof can be found in \cite[Lemma 1]{do2026optimum}, it is derived here concisely as a direct consequence of 
Lemma~\ref{lemma:convex_hull_and_support_function}:
\begin{align}
    \max_{w\in S}|w| &= \max_{w\in S}\max_{\theta\in [0,2\pi]}  \langle e^{j\theta}, w \rangle\\
    &= \max_{\theta\in [0,2\pi]} \max_{w\in S}  \langle e^{j\theta}, w \rangle\\
    &= \max_{\theta\in [0,2\pi]} \max_{w\in \conv\,S} \langle e^{j\theta}, w \rangle\\
    &= \max_{w\in \conv\,S} \max_{\theta\in [0,2\pi]} \langle e^{j\theta}, w \rangle\\
    &= \max_{w\in \conv S}|w|.
\end{align}

\section{}
\label{nadal}

Applying the Cauchy's surface area formula and Lemma~\ref{lemma:support_function_and_minkowski_sum} (presented below), we have that
\begin{align}
    &\text{perimeter of $S_1+S_2$} \nonumber \\
    &\qquad= \int_0^{2\pi} \max_{w\in S_1+S_2}\langle e^{j\theta}, w \rangle  \, d\theta \\
    &\qquad= \int_0^{2\pi} \max_{w\in S_1}\langle e^{j\theta}, w \rangle \, d\theta + \int_0^{2\pi} \max_{w\in S_2}\langle e^{j\theta}, w \rangle  \, d\theta \\
    &\qquad= \text{perimeter of $S_1$} + \text{perimeter of $S_2$}.
\end{align}

\begin{lemma}
\label{lemma:support_function_and_minkowski_sum}
For compact sets $S_1$ and $S_2$ in $\bbC$,
\begin{align}
    \max_{w\in S_1 + S_2}\langle z,w \rangle = \max_{w\in S_1}\langle z,w \rangle + \max_{w\in S_2}\langle z,w \rangle.
\end{align}
\end{lemma}
\begin{proof}
See \cite[Thm. 1.7.5]{schneider2013convex}.
\end{proof}

\section{}
\label{mac}
The proof is in essence identical to the geometric proof in Sec. \ref{geometric_proof}. An additional observation is that \eqref{minkowski_sum} has at most $N|\cW|$ vertices \cite[Ch. 13.3]{de2008computational}. Application of Lemma~\ref{lemma:maximum_perimeter_polygon} (presented below) gives
\begin{align}
    &2MN \sin \frac{\pi}{MN} \cdot g_{\cW} \nonumber \\
    &\qquad\geq \text{perimeter of } h_1 \conv\, \cW + \cdots + h_N \conv\, \cW\\
    &\qquad= \text{perimeter of }\conv\,\cW \cdot g_{\text{ideal}}.
\end{align}
Equality holds if and only if $h_1 \conv\, \cW + \cdots + h_N \conv\, \cW$ forms the regular $MN$-gon. For $\cW = \cW_M$, $h_n = e^{j2\pi\frac{n}{MN}}$ for $n=1,\ldots,N$ satisfies this condition.

\begin{lemma}
\label{lemma:maximum_perimeter_polygon}
Consider a convex polygon $S$ with $M$ vertices lying inside the unit circle. The perimeter of $S$ is at most
\begin{align}
    \text{perimeter of }\cW_M = 2M\sin\frac{\pi}{M}.
\end{align}
\end{lemma}
\begin{proof}
Although the proof is provided in \cite[Prop. 2.1]{grundbacher2025sharp}, it is repeated here for the sake of completeness.
As the statement trivially holds for $M=1$, it is assumed that $M \geq 2$. A point can then be found in $S$ that is not an extreme point. Let us think of a ray emanating from that point to every vertex. A convex polygon can be constructed by taking the convex hull of the intersection of the rays and the unit circle (see Fig. \ref{fig:better_polygon}). The perimeter of this polygon is at least that of the original polygon, from Cor. \ref{cor:perimeter_inclusion}. With this observation, it suffices to consider the polygons whose vertices are on the unit circle.
The perimeter can then be expressed as
\begin{align}
    2\sum_{m=1}^M \sin \phi_m
\end{align}
with $\sum_{m=1}^M \phi_m = \pi$ and $\phi_m \geq 0$ for all $m$. Concavity of the sine function over $[0,\pi]$ concludes the proof.
\end{proof}

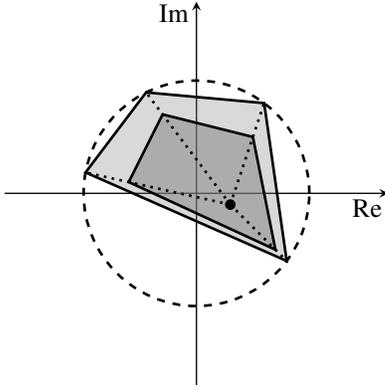
\begin{figure}
    \centering
    \begin{tikzpicture}[>=stealth]
    \begin{scope}[scale = 1.5]
        \clip (-1.7,-1.7) rectangle (1.7,1.7);
        \draw[line width=0.5pt, ->] (-1.7,0) -- (1.7,0) node (X) [right, xshift = -0.6 cm, yshift = -0.2 cm]{Re};
        \draw[line width=0.5pt, ->] (0,-1.7) -- (0,1.7) node[above,xshift = -0.3 cm, yshift = -0.4 cm]{Im};
        \coordinate (P0) at (0.3,-0.1);
        \coordinate (P1) at (0.5,0.5);
        \coordinate (P2) at (-0.3,0.7);
        \coordinate (P3) at (-0.6,0.1);
        \coordinate (P4) at (0.7,-0.5);
        \draw[dashed, line width=1pt, color=black] (0,0) circle (1cm);

        \coordinate (I1) at (0.600000,0.800000);
        \coordinate (I2) at (-0.446200,0.894933);
        \coordinate (I3) at (-0.982729,0.185051);
        \coordinate (I4) at (0.800000,-0.600000);
        
        \filldraw[draw=black, fill=gray, line width=1pt, fill opacity=0.3]
        (I1) -- (I2) -- (I3) -- (I4) -- cycle;
        \filldraw[draw=black, fill=gray, line width=1pt, fill opacity=0.5] 
        (P1) -- (P2) -- (P3) -- (P4) -- cycle;
        \filldraw[line width=1pt] (P0) circle (1pt) node{};

        \foreach \i in {1,2,3,4}{
            \draw[dotted, line width=1pt, color=black] (P0) -- (I\i);
        }
        
    \end{scope}
    \end{tikzpicture}
    \caption{Construction of a convex polygon whose perimeter is at least that of the original convex polygon.}
    \label{fig:better_polygon}
\end{figure}

\bibliographystyle{IEEEtran}
\bibliography{ref}

\end{document}